\documentclass[11pt,a4paper]{article}

\usepackage[DIV=10]{typearea} % large value = less margin (roughly speaking), recommended values: 10pt: 8, 11pt: 10, 12pt: 12

\usepackage[utf8]{inputenc}
\usepackage[T1]{fontenc}

\usepackage{microtype}

\usepackage{amsmath}
\usepackage{amssymb}
\usepackage{amsthm}

\usepackage[tt=false]{libertine}
\usepackage[varqu]{zi4}
\usepackage[libertine]{newtxmath} % must load after ams packages
\usepackage{MnSymbol}

\usepackage{tikz}

\usepackage{xcolor}
\usepackage{xspace}

\usepackage[
	style=alphabetic,
	backref=true,
	doi=false,
	url=false,
	maxcitenames=3,
	mincitenames=3,
	maxbibnames=10,
	minbibnames=10,
	backend=biber,
	sortlocale=en_US
]{biblatex}

\usepackage[ocgcolorlinks]{hyperref} % Option ocgcolorlinks makes links only colored on screen and not on printouts
\allowdisplaybreaks[1]

\renewcommand{\epsilon}{\varepsilon}

\protected\def\mathbb#1{\text{\usefont{U}{msb}{m}{n}#1}} %gives us blackboard back -> http://tex.stackexchange.com/questions/214570/try-to-use-ams-blackboard-bold-font-together-with-texgyrepagella

\makeatletter
\g@addto@macro\bfseries{\boldmath}
\makeatother

\makeatletter
\g@addto@macro\mdseries{\unboldmath}
\g@addto@macro\normalfont{\unboldmath}
\g@addto@macro\rmfamily{\unboldmath}
\g@addto@macro\upshape{\unboldmath}
\makeatother

\DeclareCiteCommand{\citem}
    {}
    {\mkbibbrackets{\bibhyperref{\usebibmacro{postnote}}}}
    {\multicitedelim}
    {}

\renewcommand*{\multicitedelim}{\addcomma\space}

\AtEveryCitekey{\clearfield{note}}

\newcommand{\myhref}[1]{%
  \iffieldundef{doi}
    {\iffieldundef{url}
       {#1}
       {\href{\strfield{url}}{#1}}}
    {\href{http://dx.doi.org/\strfield{doi}}{#1}}%
}

\DeclareFieldFormat{title}{\myhref{\mkbibemph{#1}}}
\DeclareFieldFormat
  [article,inbook,incollection,inproceedings,patent,thesis,unpublished]
  {title}{\myhref{\mkbibquote{#1\isdot}}}

\makeatletter
\AtBeginDocument{%
    \newlength{\temp@x}%
    \newlength{\temp@y}%
    \newlength{\temp@w}%
    \newlength{\temp@h}%
    \def\my@coords#1#2#3#4{%
      \setlength{\temp@x}{#1}%
      \setlength{\temp@y}{#2}%
      \setlength{\temp@w}{#3}%
      \setlength{\temp@h}{#4}%
      \adjustlengths{}%
      \my@pdfliteral{\strip@pt\temp@x\space\strip@pt\temp@y\space\strip@pt\temp@w\space\strip@pt\temp@h\space re}}%
    \ifpdf
      \typeout{In PDF mode}%
      \def\my@pdfliteral#1{\pdfliteral page{#1}}% I don't know why % this command...
      \def\adjustlengths{}%
    \fi
    \ifxetex
      \def\my@pdfliteral #1{}% isn't equivalent to this one
      \def\adjustlengths{\setlength{\temp@h}{-\temp@h}\addtolength{\temp@y}{1in}\addtolength{\temp@x}{-1in}}%
    \fi%
    \def\Hy@colorlink#1{%
      \begingroup
        \ifHy@ocgcolorlinks
          \def\Hy@ocgcolor{#1}%
          \my@pdfliteral{q}%
          \my@pdfliteral{7 Tr}% Set text mode to clipping-only
        \else
          \HyColor@UseColor#1%
        \fi
    }%
    \def\Hy@endcolorlink{%
      \ifHy@ocgcolorlinks%
        \my@pdfliteral{/OC/OCPrint BDC}%
        \my@coords{0pt}{0pt}{\pdfpagewidth}{\pdfpageheight}%
        \my@pdfliteral{F}% Fill clipping path (the url's text) with
        \my@pdfliteral{EMC/OC/OCView BDC}%
        \begingroup%
          \expandafter\HyColor@UseColor\Hy@ocgcolor%
          \my@coords{0pt}{0pt}{\pdfpagewidth}{\pdfpageheight}%
          \my@pdfliteral{F}% Fill clipping path (the url's text)
        \endgroup%
        \my@pdfliteral{EMC}%
        \my@pdfliteral{0 Tr}% Reset text to normal mode
        \my@pdfliteral{Q}%
      \fi
      \endgroup
    }%
}
\makeatother

\colorlet{DarkRed}{red!50!black}
\colorlet{DarkGreen}{green!50!black}
\colorlet{DarkBlue}{blue!50!black}

\hypersetup{
	linkcolor = DarkRed,
	citecolor = DarkGreen,
	urlcolor = DarkBlue,
	bookmarks = true,
	bookmarksnumbered = true,
	linktocpage = true
}

\newtheorem{theorem}{Theorem}[section]
\newtheorem{lemma}[theorem]{Lemma}
\newtheorem{corollary}[theorem]{Corollary}
\newtheorem{proposition}[theorem]{Proposition}

\newcommand{\dist}{\operatorname{dist}}

\newcommand{\p}{p}
\newcommand{\q}{q}

\title{A Note on Hardness of Diameter Approximation\thanks{Accepted to \emph{Information Processing Letters}. A preliminary version of this paper was presented at the \emph{31st International Symposium on Distributed Computing (DISC'17)} as a \emph{brief announcement}.}}

\author{Karl~Bringmann\thanks{Max Planck Institute for Informatics, Saarland Informatics Campus, Germany}
	\and Sebastian~Krinninger\thanks{University of Salzburg, Department of Computer Sciences, Austria. Work partially done while at Max Planck Institute for Informatics, Saarland Informatics Campus, Germany, and while at University of Vienna, Faculty of Computer Science, Austria.}
}
\date{}

\hypersetup{
	pdftitle = {A Note on Hardness of Diameter Approximation},
	pdfauthor = {Karl Bringmann, Sebastian Krinninger}
}

\begin{document}

\maketitle
\begin{abstract}
We revisit the hardness of approximating the diameter of a network.
In the CONGEST model of distributed computing, $ \tilde \Omega (n) $ rounds are necessary to compute the diameter \citem[Frischknecht et al.\ SODA'12]{FrischknechtHW12}, where $ \tilde \Omega (\cdot) $ hides polylogarithmic factors.
Abboud et al.\ \citem[DISC 2016]{AbboudCK16} extended this result to sparse graphs and, at a more fine-grained level, showed that, for any integer $ 1 \leq \ell \leq \operatorname{polylog} (n) $, distinguishing between networks of diameter $ 4 \ell + 2 $ and $ 6 \ell + 1 $ requires $ \tilde \Omega (n) $ rounds.
We slightly tighten this result by showing that even distinguishing between diameter $ 2 \ell + 1 $ and $ 3 \ell + 1 $ requires $ \tilde \Omega (n) $ rounds.
The reduction of Abboud et al.\ is inspired by recent conditional lower bounds in the RAM model, where the orthogonal vectors problem plays a pivotal role.
In our new lower bound, we make the connection to orthogonal vectors explicit, leading to a conceptually more streamlined exposition.

\end{abstract}

\section{Introduction}\label{sec:introduction}

In distributed computing, the diameter of a network is arguably the single most important quantity one wishes to compute.
In the CONGEST model~\cite{Peleg00}, where in each round every vertex can send to each of its neighbors a message of size $ O (\log{n}) $, it is known that $ \tilde \Omega(n) $ rounds are necessary to compute the diameter~\cite{FrischknechtHW12} even in sparse graphs~\cite{AbboudCK16}, where $ n $ is the number of vertices.
With this negative result in mind, it is natural that the focus has shifted towards \emph{approximating} the diameter.
In this note, we revisit hardness of computing a diameter approximation in the CONGEST model from a \emph{fine-grained} perspective.

The current fastest approximation algorithm~\cite{HolzerPRW14}, which is inspired by a corresponding RAM model algorithm~\cite{RodittyW13}, takes $ O (\sqrt{n \log n} + D) $ rounds and computes a $ \tfrac{3}{2} $-approximation of the diameter, i.e., an estimate $ \hat{D} $ such that $ \lfloor \tfrac{2}{3} D \rfloor \leq \hat{D} \leq D $, where $ D $ is the true diameter of the network.
In terms of lower bounds, Abboud, Censor-Hillel, and Khoury~\cite{AbboudCK16} showed that $ \tilde \Omega (n) $ rounds are necessary to compute a $ (\tfrac{3}{2} - \epsilon) $-approximation of the diameter for any constant $ 0 < \epsilon < \tfrac{1}{2} $.
At a more fine-grained level, they show that, for any integer $ 1 \leq \ell \leq \operatorname{polylog} (n) $, at least $ \tilde \Omega (n) $ rounds are necessary to decide whether the network has diameter $ 4 \ell + 2 $ or $ 6 \ell + 1 $, thus ruling out any ``relaxed'' notions of $(\tfrac{3}{2} -\varepsilon)$-approximation that additionally allow small additive error.
We tighten this result by showing that, for any integer $ \ell \geq 1 $, at least $ \tilde \Omega (n) $ rounds are necessary to distinguish between diameter $ 2 \ell + 1 $ and $ 3 \ell + 1 $, and more generally between diameter $ 2 \ell + \q $ and $ 3 \ell + \q $ for any $ \ell, \q \geq 1 $.

The reduction of Abboud et al.~\cite{AbboudCK16} is inspired by recent work on conditional lower bounds in the RAM model, where the \emph{orthogonal vectors problem} plays a pivotal role.
In our new lower bound, we make the connection between diameter approximation and orthogonal vectors explicit:
we consider a communication complexity version of orthogonal vectors that we show to be hard \emph{unconditionally} by a reduction from set disjointness and then devise a reduction from orthogonal vectors to diameter approximation.

Additionally, our approach has implications in the RAM model.
There, the \emph{Strong Exponential Time Hypothesis (SETH)}~\cite{ImpagliazzoPZ01} states that for every $ \delta > 0 $ there is an integer $ k \geq 3 $ such that $k$-SAT admits no algorithm with running time $ O (2^{(1 - \delta) N}) $ and the \emph{Orthogonal Vectors Hypothesis (OVH)} states that there is no algorithm to decide whether a given set of $d$-dimensional vectors of length~$ n $ contains an orthogonal pair in time $ O (n^{2 - \delta} \operatorname{poly} (d)) $ for any constant $ \delta > 0 $.
It is well-known that SETH implies OVH~\cite{Williams05}.
Prior to our work, the situation in the RAM model was as follows.
In their seminal paper~\cite{RodittyW13}, Roditty and Vassilevska Williams showed that, for any constants $ \epsilon > 0 $ and $ \delta > 0 $ there is no algorithm that computes a $ (\tfrac{3}{2} - \epsilon) $-approximation of the diameter and runs in time $ O (m^{2 - \delta}) $, unless the Strong Exponential Time Hypothesis (SETH) fails.
In particular, they show that no algorithm can decide whether a given graph has diameter $ 2 $ or $ 3 $ in time $ O (m^{2 - \delta}) $, unless the Strong Exponential Time Hypothesis (SETH) fails.
The hardness of $ 2 $ vs.\ $ 3 $ is already implied by the weaker Orthogonal Vectors Hypothesis (OVH), which in turn is implied by SETH~\cite{Williams05} and was popularized after the paper of Roditty and Vassilevska Williams appeared.
It has then been shown by Chechik et al.~\cite{ChechikLRSTW14} that, for any integer $ 1 \leq \ell \leq n^{o(1)} $, there is no algorithm that distinguishes between diameter $ 3 (\ell + 1) $ and $ 4 (\ell + 1) $ with running time $ O (m^{2 - \delta}) $ for some constant $ \delta > 0 $, unless SETH fails.
Finally, Cairo, Grossi, and Rizzi~\cite{CairoGR16} showed that, for any integer $ 1 \leq \ell \leq n^{o(1)} $, there is no algorithm that distinguishes between diameter $ 2 \ell $ and $ 3 \ell $ with running time $ O (m^{2 - \delta}) $ for some constant $ \delta > 0 $, unless SETH fails.
Our reduction reconstructs the result of Cairo et al.\ under the weaker hardness assumption OVH, yielding again a more streamlined chain of reductions.

\section{Reduction from Set Disjointness to Orthogonal Vectors}

Set disjointness is a problem in communication complexity between two players, called Alice and Bob, in which Alice is given an $n$-dimensional bit vector~$ x $ and Bob is given an $n$-dimensional bit vector~$ y $ and the goal for Alice and Bob is to find out whether there is some index $ k $ at which both vectors contain a $ 1 $, i.e., such that $ x [k] = y [k] = 1 $ (meaning the sets represented by $ x $ and~$ y $ are not disjoint).
The relevant measure in communication complexity is the number of bits exchanged by Alice and Bob in any protocol that Alice and Bob follow to determine the solution.
A classic result~\cite{KushilevitzN97,Razborov92} states that any such protocol requires Alice and Bob to exchange $ \Omega (n) $ bits to solve set disjointness.
 
In the orthogonal vectors problem, Alice is given a set of bit vectors $ L = \{ l_1, \ldots, l_n \} $ and Bob is given a set of bit vectors $ R = \{ r_1, \ldots, r_n \} $, and the goal for them is to find out if there is a pair of orthogonal vectors $ l_i \in L $ and $ r_j \in R $ (i.e., such that $ l_i [k] = 0 $ or $ r_j [k] = 0 $ in each dimension~$ k $).
We give a reduction from set disjointness to orthogonal vectors.

\begin{theorem}
Any $b$-bit protocol for the orthogonal vectors problem in which Alice and Bob each hold $ n $ vectors of dimension $ d = 2 \lceil \log n \rceil + 3 $, gives a $b$-bit protocol for the set disjointness problem where Alice and Bob each hold an $ n $-dimensional bit vector.
\end{theorem}

\begin{proof}
We show that, without any communication, Alice and Bob can transform a set disjointness instance $ \langle x, y \rangle $ with $ n $-dimensional bit vectors into an orthogonal vectors instance $ \langle L, R \rangle $ such that $ x $ and $ y $ are not disjoint if and only if $ \langle L, R \rangle $ contains an orthogonal pair.
For every integer $ 1 \leq i \leq n $, let $ s_i $ denote the binary representation of $ i $ with $ \lceil \log n \rceil $ bits.
For every bit $ b $, let $ \bar{b} $ be the result of `flipping' bit $ b $, i.e., $ \bar{1} = 0 $, and $ \bar{0} = 1 $.
Similarly, for a bit vector $ b $, let $ \bar{b} $ be the result of flipping each bit of $ b $.
For every $ 1 \leq i \leq n $, let $ l_i $ be the vector obtained from concatenating $ x [i] $, $ \bar{x} [i] $, $ \bar{x} [i] $, $ s_i $, and $ \bar{s}_i $.
For every $ 1 \leq j \leq n $, let $ r_i $ be the vector obtained from concatenating $ \bar{y} [i] $, $ y [i] $, $ \bar{y} [i] $, $ \bar{s}_i $, and $ s_i $.

We now claim that the vectors $ x $ and~$ y $ are not disjoint if and only if $ \langle L, R \rangle $ contains an orthogonal pair.
If the vectors $ x $ and $ y $ are not disjoint, then there is some $ i $ such that $ x [i] = {y [i] = 1} $.
Clearly, $ s_i $ and $ \bar{s}_i $ are orthogonal and, as the vectors $ (x [i], \bar{x} [i], \bar{x} [i]) $ and $ (\bar{y} [i], y [i], \bar{y} [i]) $ are equal to $ (1, 0, 0) $ and $ (0, 1, 0) $, respectively, they are also orthogonal.
It follows that $ l_i $ and $ r_i $ are orthogonal.

Now assume that $ \langle L, R \rangle $ contains an orthogonal pair $ l_i \in L $ and $ r_j \in R $.
We first show that $ i = j $.
Suppose for the sake of contradiction that $ i \neq j $.
Then the binary representations $ s_i $ and~$ s_j $ differ in at least one bit, say $ s_i [k] \neq s_j [k] $.
If $ s_i [k] = 0 $ and $ s_j [k] = 1 $, then $ \bar{s}_i $ and $ s_j $ are not orthogonal and thus $ l_i $ and $ r_j $ are not orthogonal, contradicting the assumption.
If $ s_i [k] = 1 $ and $ s_j [k] = 0 $, then $ s_i $ and $ \bar{s}_j $ are not orthogonal and thus $ l_i $ and $ r_j $ are not orthogonal, contradicting the assumption.
It follows that $ i = j $ and thus the vectors $ (x [i], \bar{x} [i], \bar{x} [i]) $ and $ (\bar{y} [i], y [i], \bar{y} [i]) $ are orthogonal.
Orthogonality of $ x [i] $ and $ \bar{y} [i] $ rules out $ x [i] = 1 $ and $ y [i] = 0 $, orthogonality of $ \bar{x} [i] $ and $ y [i] $ rules out $ x [i] = 0 $ and $ y [i] = 1 $, and orthogonality of $ \bar{x} [i] $ and $ \bar{y} [i] $ rules out $ x [i] = 0 $ and $ y [i] = 0 $.
It follows that $ x [i] = y [i] = 1 $, making $ x $ and $ y $ not disjoint.
\end{proof}

The hardness of set disjointness now directly transfers to orthogonal vectors.

\begin{corollary}\label{cor:communication complexity of OV}
Any protocol solving the orthogonal vectors problem with $ n $ vectors of dimension $ d = 2 \lceil \log n \rceil + 3 $, requires Alice and Bob to exchange $ \Omega (n) $ bits. 
\end{corollary}

\section{Reduction from Orthogonal Vectors to Diameter}

We now establish hardness of distinguishing between networks of diameter $ 2 \ell + \q $ and $ 3 \ell + \q $ for any $ \ell \geq 1 $ and $ \q \geq 1 $ in the CONGEST model and for any $ \ell \geq 1 $ and $ \q \geq 0 $ in the RAM model, respectively.
To unify the cases of odd and even $ \ell $, we introduce an additional parameter $ \p \in \{ 0, 1 \} $ and change the task to distinguishing between networks of diameter $ 4 \ell' - 2 \p + \q $ and $ 6 \ell' - 3 \p + \q $ for integers $ \ell' \geq 1 $, $ \q \geq 0 $, and $ \p \in \{ 0, 1 \} $.
This covers the original question: if $ \ell $ is even, then set $ \ell' := \ell/2 $ and $ \p := 0 $ and if $ \ell $ is odd, then set $ \ell' := \lceil \ell/2 \rceil $ and $ \p := 1 $.

\subsection{Construction and Implications}

Given an orthogonal vectors instance $ \langle L := \{ l_1, \ldots, l_n \}, R := \{ r_1, \ldots, r_n \} \rangle $ of $d$-dimensional vectors and parameters $ \ell \geq 1 $, $ \q \geq 0 $, and $ \p \in \{ 0, 1 \} $, we define an unweighted undirected graph $ G := G_{L, R, \ell, \p, \q} $ as follows.
The graph~$ G $ contains the following \emph{exterior} vertices: $ u_1^L, \ldots, u_n^L $, $u_1^R, \ldots, u_n^R $, $ v_1^L, \ldots, v_n^L $, $ v_1^R, \ldots, v_n^R $, $ w_1^L, \ldots, w_d^L $, $ w_1^R, \ldots, w_d^R $, $ x^L $, $ x^R $, $ y^L $, and $ y^R $.
These exterior vertices are connected by paths as follows, where each path introduces a separate set of \emph{interior} vertices:
\begin{itemize}
\item For every $ 1 \leq i \leq n $, add paths $ \pi (u_i^L, v_i^L) $ and $ \pi (u_i^R, v_i^R) $, each of length~$ \ell - \p $.
\item For every $ 1 \leq i \leq n $, add paths $ \pi (v_i^L, x^L) $ and $ \pi (v_i^R, x^R) $, each of length~$ \ell $.
\item For every $ 1 \leq i \leq n $ and every $ 1 \leq k \leq d $ such that $ l_i [k] = 1 $ add a path $ \pi (v_i^L, w_k^L) $ of length~$ \ell $.
\item For every $ 1 \leq i \leq n $ and every $ 1 \leq k \leq d $ such that $ r_i [k] = 1 $ add a path $ \pi (v_i^R, w_k^R) $ of length~$ \ell $.
\item For every $ 1 \leq k \leq d $, add paths $ \pi (y^L, w_k^L) $ and $ \pi (y^R, w_k^R) $, each of length~$ \ell $.
\item For every $ 1 \leq k \leq d $, add a path $ \pi (w_k^L, w_k^R) $ of length~$ \q $. That is, if $ \q = 0 $ then identify $ w_i^L $ and~$ w_i^R $.
\item Add a path $ \pi (x^L, y^L) $ and a path $ \pi (x^R, y^R) $, each of length~$ \ell - \p $.
\item Add a path $ \pi (y^L, y^R) $ of length~$ \p + \q $. That is, if $ \p + \q = 0 $, then identify $ y^L $ and~$ y^R $.
\end{itemize}
If $ \ell = 1 $ and $ \p = 1 $, then we identify $ x^L $ and $ y^L $, $ x^R $ and $ y^R $, $ u_i^L $ and $ v_i^L $ (for each $ 1 \leq i \leq n $), as well as $ u_i^R $ and $ v_i^R $ (for each $ 1 \leq i \leq n $), respectively.
The graph $ G $ is visualized in Figure~\ref{fig:reduction}.
Observe that $ G $ has $ O (n d \ell + d \q) $ vertices and $ O (n d \ell + d \q) $ edges.
We show that our construction has the following formal guarantees.

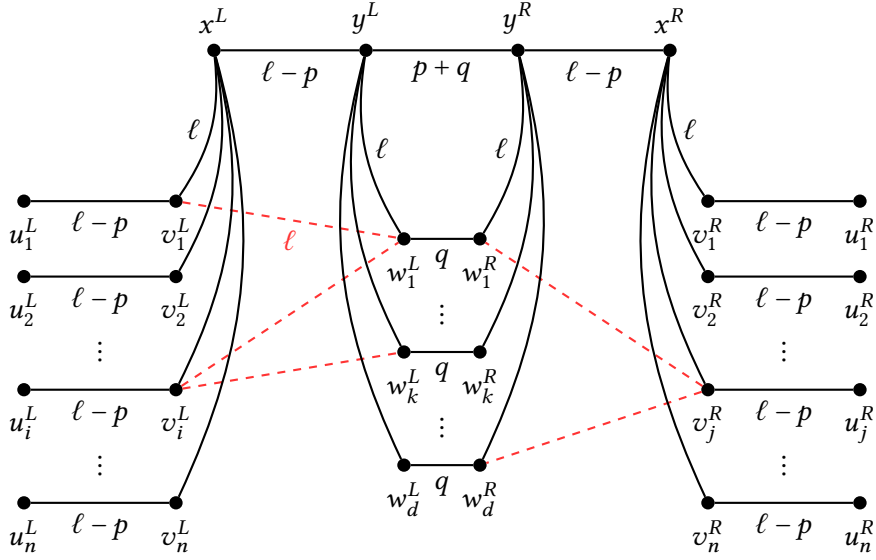
\begin{figure}
\centering

\begin{tikzpicture}
\tikzstyle{vertex}=[circle,fill=black,minimum size=5pt,inner sep=0pt,outer sep=0pt]
\tikzstyle{edge} = [thick]
\tikzstyle{gadget-edge} = [red!80,thick,dashed]

% Left nodes
\node[vertex, label=below:{$u_1^L$}] (u1L) at (-5.5, 0) {};
\node[vertex, label=below:{$u_2^L$}] (u2L) at (-5.5, -1) {};
\node[vertex, label=below:{$u_i^L$}] (uiL) at (-5.5, -2.5) {};
\node[vertex, label=below:{$u_n^L$}] (unL) at (-5.5, -4) {};

\node[vertex, label=below:{$v_1^L$}] (v1L) at (-3.5, 0) {};
\node[vertex, label=below:{$v_2^L$}] (v2L) at (-3.5, -1) {};
\node[vertex, label=below:{$v_i^L$}] (viL) at (-3.5, -2.5) {};
\node[vertex, label=below:{$v_n^L$}] (vnL) at (-3.5, -4) {};

\node at (-4.5, -2) {\vdots};
\node at (-4.5, -3.5) {\vdots};

\node[vertex, label=below:{$w_1^L$}] (w1L) at (-0.5, -0.5) {};
\node[vertex, label=below:{$w_k^L$}] (wkL) at (-0.5, -2) {};
\node[vertex, label=below:{$w_d^L$}] (wdL) at (-0.5, -3.5) {};

\node[vertex, label={$x^L$}] (xL) at (-3, 2) {};
\node[vertex, label={$y^L$}] (yL) at (-1, 2) {};

% Right nodes
\node[vertex, label=below:{$u_1^R$}] (u1R) at (5.5, 0) {};
\node[vertex, label=below:{$u_2^R$}] (u2R) at (5.5, -1) {};
\node[vertex, label=below:{$u_j^R$}] (ujR) at (5.5, -2.5) {};
\node[vertex, label=below:{$u_n^R$}] (unR) at (5.5, -4) {};

\node[vertex, label=below:{$v_1^R$}] (v1R) at (3.5, 0) {};
\node[vertex, label=below:{$v_2^R$}] (v2R) at (3.5, -1) {};
\node[vertex, label=below:{$v_j^R$}] (vjR) at (3.5, -2.5) {};
\node[vertex, label=below:{$v_n^R$}] (vnR) at (3.5, -4) {};

\node at (4.5, -2) {\vdots};
\node at (4.5, -3.5) {\vdots};

\node[vertex, label=below:{$w_1^R$}] (w1R) at (0.5, -0.5) {};
\node[vertex, label=below:{$w_k^R$}] (wkR) at (0.5, -2) {};
\node[vertex, label=below:{$w_d^R$}] (wdR) at (0.5, -3.5) {};

\node[vertex, label={$x^R$}] (xR) at (3, 2) {};
\node[vertex, label={$y^R$}] (yR) at (1, 2) {};

% Center nodes
\node at (0, -1.5) {\vdots};
\node at (0, -3) {\vdots};

% Center edges
\draw[edge] (yL) to node [below] {$\p + \q$} (yR);

\draw[edge] (w1L) to node [below] {$\q$} (w1R);
\draw[edge] (wkL) to node [below] {$\q$} (wkR);
\draw[edge] (wdL) to node [below] {$\q$} (wdR);

\draw[gadget-edge] (v1L) to node [below] {$\ell$} (w1L);
\draw[gadget-edge] (viL) to (w1L);
\draw[gadget-edge] (viL) to (wkL);
\draw[gadget-edge] (vjR) to (w1R);
\draw[gadget-edge] (vjR) to (wdR);

% Left edges
\draw[edge] (u1L) to node [below] {$\ell - \p$} (v1L);
\draw[edge] (u2L) to node [below] {$\ell - \p$} (v2L);
\draw[edge] (uiL) to node [below] {$\ell - \p$} (viL);
\draw[edge] (unL) to node [below] {$\ell - \p$} (vnL);

\draw[edge, bend left=20] (xL) to node [left] {$\ell$} (v1L);
\draw[edge, bend left=20] (xL) to (v2L);
\draw[edge, bend left=20] (xL) to (viL);
\draw[edge, bend left=20] (xL) to (vnL);

\draw[edge, bend right=20] (yL) to node [right] {$\ell$} (w1L);
\draw[edge, bend right=20] (yL) to (wkL);
\draw[edge, bend right=20] (yL) to (wdL);

\draw[edge] (xL) to node [below] {$\ell - \p$} (yL);

% Right edges
\draw[edge] (u1R) to node [below] {$\ell - \p$} (v1R);
\draw[edge] (u2R) to node [below] {$\ell - \p$} (v2R);
\draw[edge] (ujR) to node [below] {$\ell - \p$} (vjR);
\draw[edge] (unR) to node [below] {$\ell - \p$} (vnR);

\draw[edge, bend right=20] (xR) to node [right] {$\ell$} (v1R);
\draw[edge, bend right=20] (xR) to (v2R);
\draw[edge, bend right=20] (xR) to (vjR);
\draw[edge, bend right=20] (xR) to (vnR);

\draw[edge, bend left=20] (yR) to node [left] {$\ell$} (w1R);
\draw[edge, bend left=20] (yR) to (wkR);
\draw[edge, bend left=20] (yR) to (wdR);

\draw[edge] (xR) to node [below] {$\ell - \p$} (yR);
\end{tikzpicture}

\caption{Visualization of the graph $ G := G_{L, R, \ell, \p, \q} $ used in our reduction from orthogonal vectors to diameter distinction. The red, dashed edges encode the orthogonal vectors instance: For every $ 1 \leq i \leq n $ and every $ 1 \leq k \leq d $, the path $ \pi (v_i^L, w_k^L) $ is contained in $ G $ if and only if $ l_i [k] = 1 $. Similarly, for every $ 1 \leq j \leq n $ and every $ 1 \leq k \leq d $, the path $ \pi (v_i^R, w_k^R) $ is contained in $ G $ if and only if $ r_j [k] = 1 $.
}\label{fig:reduction}
\end{figure}

\begin{theorem}\label{thm:main technical diameter distinction result}
Let $ \langle L, R \rangle $ be an orthogonal vectors instance of two sets of $ d $-dimensional vectors of size $ n $ each and let $ \ell \geq 1 $, $ \p \in \{ 0, 1 \} $, and $ \q \geq 0 $ be integer parameters.
Then the unweighted, undirected graph $ G := G_{L, R, \ell, \p, \q} $ has $ O (n d \ell + d \q) $ vertices and edges and its diameter~$ D $ has the following property: if $ \langle L, R \rangle $ contains an orthogonal pair, then $ D = 6 \ell - 3 \p + \q $, and if $ \langle L, R \rangle $ contains no orthogonal pair, then $ D = 4 \ell - 2 \p + \q $.
\end{theorem}

Before we give a proof of this statement, we motivate it by discussing its immediate consequences in the CONGEST model and the RAM model.
For the CONGEST model, observe that $ G $ has a small cut of size $ d + 1 $ between its left hand side and its right hand side.
A standard simulation argument, where communication between Alice and Bob is limited to messages sent along the small cut, yields our main result.

\begin{corollary}
In the CONGEST model, any algorithm distinguishing between graphs of diameter~$ 2 \ell + \q $ and graphs of diameter~$ 3 \ell + \q $ when $ \ell \geq 1 $ and $ \q \geq 1 $ requires $ \Omega (n / ((\ell + \q) \log^3{n})) $ rounds.
\end{corollary}

\begin{proof}
Let $ \langle L, R \rangle $ be an orthogonal vectors instance with $ n $ vectors of dimension $ d = 2 \lceil \log n \rceil + 3 $ and let $ \mathcal{A} $ be an algorithm distinguishing between graphs of diameter~$ 2 \ell + 1 $ and graphs of diameter~$ 3 \ell + 1 $.
If $ \ell $ is even, then set $ \ell' := \ell/2 $ and $ \p := 0 $ and if $ \ell $ is odd, then set $ \ell' := \lceil \ell/2 \rceil $ and $ \p := 1 $.
Then by Theorem~\ref{thm:main technical diameter distinction result} the graph $ G := G_{L, R, \ell', \p, \q} $ has diameter $ 3 \ell + \q $ if $ \langle L, R \rangle $ contains an orthogonal pair and $ 2 \ell + \q $ otherwise.
Observe that $ G $ has $ n' = O (n (\ell+q) \log{n}) $ edges and since $ \q \geq 1 $ it can be partitioned into two node sets $ A $ and $ B $ such that
\begin{itemize}
\item $ G [A] $, the subgraph of $ G $ induced by $ A $, is fully determined by $ L $, $ \ell' $, $ \p $, and $ \q $.
\item $ G [B] $, the subgraph of $ G $ induced by $ B $, is fully determined by $ R $, $ \ell' $, $ \p $, and $ \q $.
\item The number of edges between $ A $ and $ B $ in $ G $ is $ d + 1 = O (\log{n}) $.
\end{itemize}
Thus, Alice and Bob can simulate running $ \mathcal{A} $ on the graph $ G $ as follows:
Alice constructs the graph $ G [A] $ and simulates the states of all vertices in $ A $ as well as the messages sent between them and Bob constructs the graph $ G [B] $ and simulates the states of all vertices in $ B $ as well as the messages sent between them.
Every time a message is sent from a node in $ A $ to a node in $ B $, Alice communicates the $ O (\log n) $-size message to Bob and every time a message is sent from a node in $ B $ to a node in $ A $, Bob communicates the $ O (\log n) $-size message to Alice.
Since $ \q \geq 1 $, the subgraph $ G [A] $ simulated by Alice is separated by $ O (\log{n}) $ edges from the subgraph $ G [B] $ simulated by Bob.
Thus, in each simulated round of $ \mathcal{A} $ at most $ O (\log^2 n) $ bits can be sent from Alice to Bob and vice versa.
As Alice and Bob need to exchange $ \Omega (n) $ bits to determine the result to the orthogonal vectors problem by Corollary~\ref{cor:communication complexity of OV}, the algorithm $ \mathcal{A} $ requires $ \Omega (n / \log^2{n}) = \Omega (n' / ((\ell + \q) \log^3{n'})) $ rounds.
\end{proof}

In the RAM model, the Orthogonal Vectors Hypothesis (OVH) states that there is no algorithm that decides whether a given orthogonal vectors instance contains an orthogonal pair in time $ O (n^{2 - \delta} \operatorname{poly} (d)) $ for some constant $ \delta > 0 $.\footnote{In Section~\ref{sec:introduction}, we have mentioned a variant of the orthogonal vectors problem with a single input set. By a straightforward reduction, this variant has the same asymptotic time complexity as the variant with two input sets defined above.}
Under this hardness assumption, our reduction has the following straightforward implication.

\begin{corollary}
In the RAM model, under OVH, there is no algorithm distinguishing between graphs of diameter~$ 2 \ell + \q $ and graphs of diameter~$ 3 \ell + \q $, where $ \ell \geq 1 $ and $ \q \geq 0 $, in time $ O (m^{2 - \delta} / (\ell + \q)^{2 - \delta}) $ for any constant $ \delta > 0 $.
\end{corollary}

\begin{proof}
Let $ \langle L, R \rangle $ be an orthogonal vectors instance with $ n $ vectors of dimension $ d $ and let $ \mathcal{A} $ be an algorithm distinguishing between graphs of diameter~$ 2 \ell + \q $ and graphs of diameter~$ 3 \ell + \q $ running in time $ O (m^{2 - \delta} / (\ell + \q)^{2 - \delta}) $.
If $ \ell $ is even, then set $ \ell' := \ell/2 $ and $ \p := 0 $ and if $ \ell $ is odd, then set $ \ell' := \lceil \ell/2 \rceil $ and $ \p := 1 $.
Then by Theorem~\ref{thm:main technical diameter distinction result} the graph $ G := G_{L, R, \ell', \p, \q} $ has diameter $ 3 \ell + \q $ if $ \langle L, R \rangle $ contains an orthogonal pair and $ 2 \ell + \q  $ otherwise.
Observe that $ G $ has $ m = O (n d \ell + d \q) $ edges and thus $ \mathcal{A} $ will take time $ O (n^{2 - \delta} d^{2 - \delta}) $ on $ G $.
This yields an algorithm solving any orthogonal vectors instance in time $ O (n^{2 - \delta} \operatorname{poly} (d)) $, contradicting OVH.
\end{proof}

\subsection{Proof of Theorem~\ref{thm:main technical diameter distinction result}}\label{sec:diameter distinction}

Before we give the proof of Theorem~\ref{thm:main technical diameter distinction result}, we introduce the following useful terminology:
For every $ 1 \leq i \leq n $, $ P_i^L $ is defined as the set of all vertices that lie on one of the following paths: $ \pi (u_i^L, v_i^L) $, $ \pi (v_i^L, y^L) $ (\emph{excluding $ y^L $}), or $ \pi (v_i^L, w_k^L) $ (\emph{excluding $ w_k^L $}) for some $ 1 \leq k \leq d $ such that $ l_i [k] = 1 $.
Similarly, for every $ 1 \leq i \leq n $, $ P_i^R $ is defined as the set of all vertices that lie on one of the following paths: $ \pi (u_i^R, v_i^R) $, $ \pi (v_i^R, y^R) $ (\emph{excluding $ y^R $}), or $ \pi (v_i^R, w_k^R) $ (\emph{excluding $ w_k^R $}) for some $ 1 \leq k \leq d $ such that $ r_i [k] = 1 $.
We set $ V^L := \bigcup_{1 \leq i \leq n} P_i^R $ (left vertices),  $ V^R := \bigcup_{1 \leq i \leq n} P_i^R $ (right vertices), and $ V^M := V \setminus (V^L \cup V^R) $.
Note that $ V^M $ consists on all vertices that lie on $ \pi (y^L, y^R) $, $ \pi (x^L, y^L) $, $ \pi (x^R, y^R) $, $ \pi (y^L, w_k^L) $ for some $ 1 \leq k \leq d $, $ \pi (y^R, w_k^R) $ for some $ 1 \leq k \leq d $, or $ \pi (w_k^L, w_k^R) $ for some $ 1 \leq k \leq d $.

We now state some universal upper and lower bounds on distances in the graph~$ G $ that hold regardless of whether the orthogonal vectors instance contains an orthogonal pair.
Their correctness can readily be verified and we also give rigorous proofs in the appendix.

\begin{lemma}\label{lem:distance to v}
For every orthogonal vectors instance, $ \dist_G (s, v_i^L) \leq \ell - \p $ for every $ 1 \leq i \leq n $ and every $ s \in P_i^L $ and $ \dist_G (v_j^R, t) \leq \ell - \p $ for every $ 1 \leq j \leq n $ and every $ t \in P_j^R $.
\end{lemma}

Note that in Lemma~\ref{lem:distance to v} it is crucial that we have defined each path $ P_i^L $ to exclude the node~$ x^L $ and that $ \p \leq 1 $.

\begin{lemma}\label{lem:distance upper middle vertices}
For every orthogonal vectors instance and every pair of vertices $ s, t \in V^M $, $ \dist_G (s, t) \leq 4 \ell - 2 \p + \q $ and more specifically $ \dist_G (x^L, v) \leq 2 \ell - \p + \q $ and $ \dist_G (v, x^R) \leq 2 \ell - \p + \q $ for every vertex $ v \in V^M $.
\end{lemma}

\begin{lemma}\label{lem:distance upper bounds}
For every orthogonal vectors instance and every pair of vertices $ s $ and $ t $, if $ s \in V^L \cup V^M $ and $ t \in V^L \cup V^M $ or $ s \in V^R \cup V^M $ and $ t \in V^R \cup V^M $, then $ \dist_G (s, t) \leq 4 \ell - 2 \p + \q $.
\end{lemma}

\begin{lemma}\label{lem:distance lower bounds}
For every orthogonal vectors instance, the following holds in $ G $:
\begin{itemize}
\item $ \dist_G (u_i^L, y^L) = 3 \ell - 2 \p $ and $ \dist_G (u_i^R, y^R) = 3 \ell - 2 \p $ for every $ 1 \leq i \leq n $,
\item $ \dist_G (u_i^L, v_{i'}^L) = 3 \ell - \p $ and $ \dist_G (u_i^R, v_{i'}^R) = 3 \ell - \p $ for every $ 1 \leq i \leq n $ and every $ 1 \leq i' \leq n $ such that $ i' \neq i $,
\item $ \dist_G (u_i^L, w_k^L) \geq 2 \ell - \p $ and $ \dist_G (u_i^R, w_k^R) \geq 2 \ell - \p $ for every $ 1 \leq i \leq n $ and every $ 1 \leq k \leq n $.
\end{itemize}
\end{lemma}

We finally give the proof of Theorem~\ref{thm:main technical diameter distinction result}.
We split up the two cases (containing an orthogonal pair or not) into two pieces, whose proofs follow a similar pattern.

\begin{proposition}
If the orthogonal vectors instance $ \langle L, R \rangle $ contains no orthogonal pair, then $ D = 4 \ell - 2 \p + \q $.
\end{proposition}

\begin{proof}
We first show that $ D \leq 4 \ell - 2 \p + \q $, i.e., $ \dist_G (s, t) \leq 4 \ell - 2 \p + \q $ for every pair of vertices $ s, t \in V $.
Note that by Lemma~\ref{lem:distance upper bounds} we only have to show that $ \dist_G (s, t) \leq 4 \ell - 2 \p + \q $ whenever $ s \in P_i^L $ for some $ 1 \leq i \leq n $ and $ t \in P_j^R $ for some $ 1 \leq j \leq n $.
By Lemma~\ref{lem:distance to v} we have $ \dist_G (s, v_i^L) \leq \ell - \p $ and $ \dist_G (v_j^R, t) \leq \ell - \p $ for such $ s \in P_i^L $ and $ t \in P_j^R $.
Since the orthogonal vectors instance contains no orthogonal pair there is a $ 1 \leq k \leq d $ such that both $ l_i [k] = r_j [k] = 1 $.
Thus, our graph $ G $ contains both paths $ \pi (v_i^L, w_k^L) $ and $ \pi (w_k^R, v_j^R) $, each of length $ \ell $.
By the triangle inequality we therefore have
\begin{align*}
\dist_G (s, t) \leq \underbrace{\dist_G (s, v_i^L)}_{\leq \ell - \p \text{ (Lemma~\ref{lem:distance to v})}} &+ \underbrace{\dist_G (v_i^L, w_k^L)}_{\leq \ell} + \underbrace{\dist_G (w_k^L, w_k^R)}_{\leq \q} \\
 &+ \underbrace{\dist_G (w_k^R, v_j^R)}_{\leq \ell} + \underbrace{\dist_G (v_j^R, t)}_{\leq \ell - \p \text{ (Lemma~\ref{lem:distance to v})}} \leq 4 \ell - 2 \p + \q \, .
\end{align*}

It remains to show that $ D \geq 4 \ell - 2 \p + \q $.
We will argue that $ \dist_G (u_1^L, u_1^R) \geq 4 \ell - 2 \p + \q $.
Since the paths $ \pi (y^L, y^R), \pi (w_1^L, w_1^R), \ldots, \pi (w_d^L, w_d^R) $ separate the left part of the graph from the right part of the graph, every path from $ u_i^L $ to $ u_j^R $ must contain at least one of these paths entirely.
If the shortest path from $ u_i^L $ to $ u_j^R $ contains the path $ \pi (y^L, y^R) $ entirely, then, since $ \dist_G (u_i^L, y^L) = 3 \ell - 2 \p $ and $ \dist_G (u_j^R, y^R) = 3 \ell - 2 \p $ by Lemma~\ref{lem:distance lower bounds},
\begin{equation*}
\dist_G (u_1^L, u_1^R) = \underbrace{\dist_G (u_1^L, y^L)}_{\geq 3 \ell - 2 \p \text{ (Lemma~\ref{lem:distance lower bounds})}} + \underbrace{| \pi (y^L, y^R) |}_{= \p + \q} + \underbrace{\dist_G (y^R, u_1^R)}_{\geq 3 \ell - 2 \p \text{ (Lemma~\ref{lem:distance lower bounds})}} \geq 6 \ell - 3 \p + \q \, .
\end{equation*}
If the shortest path from $ u_i^L $ to $ u_j^R $ contains the path $ \pi (w_k^L, w_k^R) $ for some $ 1 \leq k \leq d $ entirely, then the argument is as follows:
By Lemma~\ref{lem:distance lower bounds} we have $ \dist_G (u_1^L, w_k^L) \geq 2 \ell - \p $ and $ \dist_G (w_k^R, u_1^R) \geq 2 \ell - \p $.
We therefore get
\begin{equation*}
\dist_G (u_1^L, u_1^R) = \underbrace{\dist_G (u_1^L, w_k^L)}_{\geq 2 \ell - \p \text{ (Lemma~\ref{lem:distance lower bounds})}} + \underbrace{| \pi (w_k^L, w_k^R) |}_{= \q} + \underbrace{\dist_G (w_k^R, u_1^R)}_{\geq 2 \ell - \p \text{ (Lemma~\ref{lem:distance lower bounds})}} \geq 4 \ell - 2 \p + \q \, .
\end{equation*}
\end{proof}

\begin{proposition}
If the orthogonal vectors instance $ \langle L, R \rangle $ contains an orthogonal pair, then $ D = 6 \ell - 3 \p + \q $.
\end{proposition}

\begin{proof}
We first show that $ D \geq 6 \ell - 3 \p + \q $.
Let $ l_i \in L $ and $ r_j \in R $ denote the orthogonal pair.
We will argue that $ \dist_G (u_i^L, u_j^R) \geq 6 \ell - 3 \p + \q $.

Since the paths $ \pi (y^L, y^R), \pi (w_1^L, w_1^R), \ldots, \pi (w_d^L, w_d^R) $ separate the left part of the graph from the right part of the graph, every path from $ u_i^L $ to $ u_j^R $ must contain at least one of these paths entirely.
If the shortest path from $ u_i^L $ to $ u_j^R $ contains the path $ \pi (y^L, y^R) $ entirely, then, since $ \dist_G (u_i^L, y^L) = 3 \ell - 2 \p $ and $ \dist_G (u_j^R, y^R) = 3 \ell - 2 \p $ by Lemma~\ref{lem:distance lower bounds},
\begin{equation*}
\dist_G (u_i^L, u_j^R) = \underbrace{\dist_G (u_i^L, y^L)}_{\geq 3 \ell - 2 \p \text{ (Lemma~\ref{lem:distance lower bounds})}} + \underbrace{| \pi (y^L, y^R) |}_{= \p + \q} + \underbrace{\dist_G (y^R, u_j^R)}_{\geq 3 \ell - 2 \p \text{ (Lemma~\ref{lem:distance lower bounds})}} \geq 6 \ell - 3 \p + \q \, .
\end{equation*}
If the shortest path from $ u_i^L $ to $ u_j^R $ contains the path $ \pi (w_k^L, w_k^R) $ for some $ 1 \leq k \leq d $ entirely, then the argument is as follows.
Since $ l $~and~$ r $ are an orthogonal pair, we have $ r_j [k] = 0 $ or $ l_i [k] = 0 $.
By symmetry, assume $ r_j [k] = 0 $, which implies that the path $ \pi (w_k^R, v_j^R) $ is not contained in~$ G $.
Since the shortest path is simple, it is either the case that after the vertex $ w_k $ the shortest path contains (a) the subpath $ \pi (w_k^R, y^R) $ or (b) the subpath $ \pi (w_k^R, v_{j'}^R) $ for some $ 1 \leq j' \leq n $ such that $ j' \neq j $.
By Lemma~\ref{lem:distance lower bounds} we have $ \dist_G (u_j^R, w_k^L) \geq 2 \ell - \p $.
In case (a) we additionally use $ \dist_G (u_j^R, y^R) = 3 \ell - 2 \p $ from Lemma~\ref{lem:distance lower bounds} and thus get
\begin{align*}
\dist_G (u_i^L, u_j^R) & = \underbrace{\dist_G (u_i^L, w_k^L)}_{\geq 2 \ell - \p \text{ (Lemma~\ref{lem:distance lower bounds})}} + \underbrace{| \pi (w_k^L, w_k^R) |}_{= \q} + \underbrace{| \pi (w_k^R, y^R) |}_{= \ell} + \underbrace{\dist_G (y^R, u_j^R)}_{\geq 3 \ell - 2 \p \text{ (Lemma~\ref{lem:distance lower bounds})}} \\
 & \geq 6 \ell - 2 \p + \q \geq 6 \ell - 3 \p + \q \, .
\end{align*}
In case (b) we additionally use $ \dist_G (u_j^R, v_{j'}^R) = 3 \ell - \p $ from Lemma~\ref{lem:distance lower bounds} and thus get
\begin{align*}
\dist_G (u_i^L, u_j^R) & = \underbrace{\dist_G (u_i^L, w_k^L)}_{\geq 2 \ell - \p \text{ (Lemma~\ref{lem:distance lower bounds})}} + \underbrace{| \pi (w_k^L, w_k^R) |}_{= \q} + \underbrace{| \pi (w_k^R, v_{j'}) |}_{= \ell} + \underbrace{\dist_G (v_{j'}, u_j^R)}_{\geq 3 \ell - \p \text{ (Lemma~\ref{lem:distance lower bounds})}} \\
 & \geq 6 \ell - \p + \q \geq 6 \ell - 3 \p + \q \, .
\end{align*}

It remains to show that $ D \leq 6 \ell - 3 \p + \q $.
By Lemma~\ref{lem:distance upper bounds} and since $ 4 \ell - 2 \p + \q \leq 6 \ell - 3 \p + \q $, we only have to show that $ \dist_G (s, t) \leq 6 \ell - 3 \p + \q $ when $ s \in P_i^L $ for some $ 1 \leq i \leq n $ and $ t \in P_j^R $ for some $ 1 \leq j \leq n $.
By Lemma~\ref{lem:distance to v} we have $ \dist_G (s, v_i^L) \leq \ell - \p $ and $ \dist_G (v_j^R, t) \leq \ell - \p $ for such $ s \in P_i^L $ and $ t \in P_j^R $.
By the triangle inequality we therefore have
\begin{align*}
\dist_G (s, t) \leq &\underbrace{\dist_G (s, v_i^L)}_{\leq \ell - \p \text{ (Lemma~\ref{lem:distance to v})}} + \underbrace{\dist_G (v_i^L, x^L)}_{\leq \ell} + \underbrace{\dist_G (x^L, y^L)}_{\leq \ell - \p} + \underbrace{\dist_G (y^L, y^R)}_{\leq \p + \q} \\
 &+ \underbrace{\dist_G (y^R, x^R)}_{\leq \ell - \p} + \underbrace{\dist_G (x^R, v_j^R)}_{\leq \ell} + \underbrace{\dist_G (v_j^R, t)}_{\leq \ell - \p \text{ (Lemma~\ref{lem:distance to v})}} \leq 6 \ell - 3 \p + \q \, .
\end{align*}
\end{proof}

\printbibliography[heading=bibintoc] % Make bibliography show up in table of contents

\appendix
\section*{Appendix}
In this appendix, we provide rigorous proofs of Lemmas~\ref{lem:distance to v} to~\ref{lem:distance lower bounds}.

\subsection*{Proof of Lemma~\ref{lem:distance to v}}

We only proof the first part of the claim.
The second part then follows from symmetric arguments.
There are three possibilities for a vertex $ s $ to be contained in~$ P_i^L $:
\begin{enumerate}
\item $ s $ lies on the path $ \pi (u_i^L, v_i^L) $ (which has length $ \ell - \p $)
\item $ s $ lies on the path $ \pi (v_i^L, x^L) $ (which has length $ \ell $) and $ s \neq x^L $
\item $ s $ lies on the path $ \pi (v_i^L, w_k^L) $ (which has length $ \ell $) for some $ 1 \leq k \leq d $ such that $ l_i [k] = 1 $ and $ s \neq w_k^L $
\end{enumerate}
As $ p \leq 1 $, we have $ \ell - 1 \leq \ell - \p $ and thus in each of the three cases we have $ \dist_G (s, v_i^L) \leq \ell - \p $.

\subsection*{Proof of Lemma~\ref{lem:distance upper middle vertices}}

Consider all simple paths from $ x^L $ to $ x^R $ in $G[V^M]$:
These are $ \pi (x^L, y^L), \pi(y^L, y^R), \pi(y^R, x^R) $ as well as $ \pi(x^L, y^L), \pi(y^L, w_k^L), \pi(w_k^L, w_k^R), \pi(w_k^R, y^R), \pi(y^R, x^R) $ for every $ 1 \leq k \leq d $.
These paths have length $ 2 (\ell - \p) + \p + \q \leq 4 \ell - 2 \p + \q $ and $ 2 (\ell - \p) + 2 \ell + \q =
4 \ell - 2 \p + \q $, respectively, in both cases we obtain length $ \leq 4 \ell - 2 \p + \q $.
Moreover, each node in $ V^M $ is contained in (at least) one of these paths.
From these paths, pick $ P_s $, $ P_t $ containing $ s, t $.
Following $ P_s $ and then following the reversed $ P_t $ yields a cyclic walk from $ x^L $ to itself containing $ s $ and $ t $.
This walk has length at most $ 2 (4 \ell - 2 \p + \q) $, and thus the induced walk from $ s $ to $ t $ or the one from $ t $ to $ s $ has length at most $ 4 \ell - 2 \p + \q $. This yields $\dist_G(s,t) \le 4 \ell - 2 \p + \q $.

For the stronger bound $ \dist_G (x^L, v) \leq 2 \ell - \p + \q $ for every vertex $ v \in V^M $, consider the following paths:
\begin{itemize}
\item $ \pi(x^L, y^L), \pi(y^L, w_k^L), \pi(w_k^L, w_k^R) $ for any $ 1 \leq k \leq d $ has length $ 2 \ell - \p + \q $
\item $ \pi(x^L, y^L), \pi(y^L, y^R), \pi(y^R ,w_k^R) $ minus the last vertex, for any $ 1 \leq k \leq d $, has length $ (\ell - \p) + (\p + \q) + \ell - 1 \leq (\ell - \p) + (\p + \q) + (\ell - \p) = 2 \ell - \p + \q $.
\item $ \pi(x^L, y^L), \pi(y^L, y^R), \pi(y^R, x^R) $ has length $2 (\ell - \p) + (\p + \q) = 2 \ell - \p + \q $
\end{itemize}
Since these paths cover all vertices in $ V^M $, each vertex in $ V^M $ has distance at most $ 2 \ell - \p + \q $ to~$ x^L $.
An analogous argument gives $ \dist_G (v, x^R) \leq 2 \ell - \p + \q $ for every vertex $ v \in V^M $.

\subsection*{Proof of Lemma~\ref{lem:distance upper bounds}}

By Lemma~\ref{lem:distance upper middle vertices} we have $ \dist_G (s, t) \leq 4 \ell - 2 \p + \q $ for $ s, t \in V^M $.
Next, consider the case $ s \in V^L $ and $ t \in V^L $, say $ s \in P_i^L $ for some $ 1 \leq i \leq n $ and $ t \in P_j^L $ for some $ 1 \leq j \leq n $.
Then we have
\begin{align*}
\dist_G (s, t) & \leq \underbrace{\dist_G (s, v_i^L)}_{\leq \ell - \p \text{ by Lemma~\ref{lem:distance to v}}} + \underbrace{\dist_G (v_i^L, x^L)}_{\leq | \pi (x^L, v_i^L) | \leq \ell} + \underbrace{\dist_G (x^L, v_j^L)}_{\leq | \pi (x^L, v_i^L) | \leq \ell} + \underbrace{\dist_G (v_j^L, t)}_{\leq \ell - \p \text{ by Lemma~\ref{lem:distance to v}}} \\
	& \leq 4 \ell - 2 \p \, .
\end{align*}
Finally, consider the case $ s \in V^L $ and $ t \in V^M $, say $ s \in P_i^L $ for some $ 1 \leq i \leq n $.
By Lemma~\ref{lem:distance upper middle vertices} we have $ \dist_G (x^L, t) \leq 2 \ell - \p + \q $ and thus get
\begin{equation*}
\dist_G (s, t) \leq \underbrace{\dist_G (s, v_i^L)}_{\leq \ell - \p \text{ by Lemma~\ref{lem:distance to v}}} + \underbrace{\dist_G (v_i^L, x^L)}_{\leq | \pi (x^L, v_i^L) | \leq \ell} + \underbrace{\dist_G (x^L, t)}_{\leq 2 \ell - \p + \q \text{ by Lemma~\ref{lem:distance upper middle vertices}}} \leq 4 \ell - 2 \p + \q \, .
\end{equation*}
The remaining cases require symmetric arguments in which the roles of $ L $ and $ R $ are exchanged.

\subsection*{Proof of Lemma~\ref{lem:distance lower bounds}}

Let $ 1 \leq i \leq n $.
Observe that all simple paths of length at most $ 3 \ell - \p $ starting at vertex $ u_i^L $ and ending at an exterior vertex must be of the following form (where some of these paths are of length at most $ 3 \ell - \p $ only if $ \p = 0 $, and some only if $ \p = 0 $ and $ \q = 0 $):
\begin{enumerate}
\item $ \pi (u_i^L, v_i^L) $ (of length $ \ell - \p $),
\item $ \pi (u_i^L, v_i^L), \pi (v_i^L, x^L) $ (of length $ 2 \ell - \p $), \label{item:contains xL 1}
\item $ \pi (u_i^L, v_i^L), \pi (v_i^L, x^L), \pi (x^L, y^L) $ (of length $ 3 \ell - 2 \p $),
\item $ \pi (u_i^L, v_i^L), \pi (v_i^L, x^L), \pi (x^L, y^L), \pi (y^L, y^R) $ (of length $ 3 \ell - \p + \q $),
\item $ \pi (u_i^L, v_i^L), \pi (v_i^L, x^L), \pi (x^L, v_{i'}^L) $ for some $ 1 \leq i' \leq n $ (of length $ 3 \ell - \p $),
\item $ \pi (u_i^L, v_i^L), \pi (v_i^L, w_k^L) $ for some $ 1 \leq k \leq d $ such that $ l_i [k] = 1 $ (of length $ 2 \ell - \p $),
\item $ \pi (u_i^L, v_i^L), \pi (v_i^L, w_k^L), \pi (w_k^L, y^L) $ for some $ 1 \leq k \leq d $ such that $ l_i [k] = 1 $ (of length $ 3 \ell - \p $),
\item $ \pi (u_i^L, v_i^L), \pi (v_i^L, w_k^L), \pi (w_k^L, y^L), \pi (y^L, y^R) $ for some $ 1 \leq k \leq d $ such that $ l_i [k] = 1 $ (of length $ 3 \ell + \q $),
\item $ \pi (u_i^L, v_i^L), \pi (v_i^L, w_k^L), \pi (w_k^L, v_{i'}^L) $ for some $ 1 \leq k \leq d $ and some $ 1 \leq i' \leq n $ with $ i' \neq i $ such that $ l_i [k] = 1 $ and $ l_{i'} [k] = 1 $ (of length $ 3 \ell - \p $),
\item $ \pi (u_i^L, v_i^L), \pi (v_i^L, w_k^L), \pi (w_k^L, w_k^R) $ for some $ 1 \leq k \leq d $ (of length $ 2 \ell - \p + \q $),
\item $ \pi (u_i^L, v_i^L), \pi (v_i^L, w_k^L), \pi (w_k^L, w_k^R), \pi (w_k^R, y^R) $ for some $ 1 \leq k \leq d $ (of length $ 3 \ell - \p + \q $),
\item $ \pi (u_i^L, v_i^L), \pi (v_i^L, w_k^L), \pi (w_k^L, w_k^R), \pi (w_k^R, y^R), \pi (y^R, y^L) $ for some $ 1 \leq k \leq d $ (of length $ 3 \ell + 2 \q $), or
\item $ \pi (u_i^L, v_i^L), \pi (v_i^L, w_k^L), \pi (w_k^L, w_k^R), \pi (w_k^R, v_{j'}^R) $ for some $ 1 \leq k \leq d $ and some $ 1 \leq j' \leq n $ such that $ l_i [k] = 1 $ and $ r_{j'} [k] = 1 $ (of length $ 3 \ell - \p + \q $).
\end{enumerate}
It follows that $ \dist_G (u_i^L, v_{i'}^L) = 3 \ell - \p $ (as all paths of length at most $ 3 \ell - 2 \p $ ending at $ v_{i'}^L $ have length $ 3 \ell - \p $), $ \dist_G (u_i^L, y^L)= 3 \ell - 2 \p $ (as the shortest path of length at most $ 3 \ell - 2 \p $ ending at $ y^L $ has length $ 3 \ell - 2 \p $), and $ \dist_G (u_i^L, w_k^L) \geq 2 \ell - \p $ for every $ 1 \leq k \leq d $ (as the only possible path of length at most $ 3 \ell - 2 \p $ ending at $ w_k^L $ has length $ 2 \ell - \p $).
A symmetric argument gives $ \dist_G (u_i^R, v_{i'}^R) = 3 \ell - \p $, $ \dist_G (u_i^R, y^R) = 3 \ell - 2 \p $, and $ \dist_G (u_i^R, w_k^R) \geq 2 \ell - \p $ for every $ 1 \leq k \leq d $.

\end{document}